\documentclass[journal,onecolumn]{IEEEtran}
\usepackage{amssymb}
\usepackage{caption2}
\usepackage{amsmath}
\usepackage{multirow}
\usepackage{amsthm}
\usepackage{graphics}
\usepackage{color}

\newtheorem{theorem}{Theorem}[section]
\newtheorem{lemma}[theorem]{Lemma}

\newtheorem{corollary}[theorem]{Corollary}

\newtheorem{remark}[theorem]{Remark}

\newtheorem{example}[theorem]{Example}

\newcommand{\ma}{\mathcal}

\newcommand{\s}{\subseteq}

\newcommand{\fr}{\frac}

\begin{document}
\title{Centralized coded caching schemes: \\A hypergraph theoretical approach}

\author{Chong Shangguan, Yiwei Zhang and Gennian Ge
\thanks{The research of G. Ge was supported by the National Natural Science Foundation of China under Grant Nos. 11431003 and 61571310.}
\thanks{C. Shangguan is with the School of Mathematical Sciences, Zhejiang University, Hangzhou 310027, China (email: 11235061@zju.edu.cn).}
\thanks{Y. Zhang is with the School of Mathematical Sciences, Capital Normal University, Beijing 100048, China (email: rexzyw@163.com).}
\thanks{G. Ge is with the School of Mathematical Sciences, Capital Normal University,
Beijing 100048, China. He is also with Beijing Center for Mathematics and Information Interdisciplinary Sciences, Beijing 100048, China (e-mail: gnge@zju.edu.cn).}
}

\date{}\maketitle

\begin{abstract}

The centralized coded caching scheme is a technique proposed by Maddah-Ali and Niesen as a solution to reduce the network burden in peak times in a wireless network system. Yan et al. reformulate the problem as designing a corresponding placement delivery array and propose two new schemes from this perspective. These schemes above significantly reduce the transmission rate $R$, compared with the uncoded caching scheme. However, to implement the new schemes, each file should be cut into $F$ pieces, where $F$ increases exponentially with the number of users $K$. Such constraint is obviously infeasible in the practical setting, especially when $K$ is large. Thus it is desirable to design caching schemes with constant rate $R$ (independent of $K$) as well as small $F$.

In this paper we view the centralized coded caching problem in a hypergraph perspective and show that designing a feasible placement delivery array is equivalent to constructing a linear and (6, 3)-free 3-uniform 3-partite hypergraph. Several new results and constructions arise from our novel point of view.
First, by using the famous (6, 3)-theorem in extremal combinatorics, we show that constant rate caching schemes with $F$ growing linearly with $K$ do not exist. Second, we present two infinite classes of centralized coded caching schemes, which include the schemes of Ali-Niesen and Yan et al. as special cases, respectively. Moreover, our constructions show that constant rate caching schemes with $F$ growing sub-exponentially with $K$ do exist.
\end{abstract}


\section{Introduction}

Video delivery has become the main driving factor for the wireless data traffic in our daily life and it is faced with a dramatic increasing demand \cite{cisco}. Suppose that we have a server with a large library of contents connecting to a group of users. At certain times each user may demand a specific file from the server. Excessive demands at the same time would often jam the wireless traffic, leading to delays and overloads in the system and then poor user experience. Therefore, there have been great interests from both academia and industry to solve this problem. The solution is to take advantages of the memories distributed across the network, especially those close to the end users, to duplicate contents. We call the duplication of the contents as {\it caches}. The system allocates some fractions of the contents into the cache of each user when the network load is low, thus in peak times user requests can be served from these caches. In this manner we can reduce the network burden and smooth the network traffic.

In their seminal work on coded caching schemes \cite{AN}, Maddah-Ali and Niesen propose the {\it centralized coded caching scheme} (or CCC scheme for simplicity), where the term ``centralized" means that we only have one server in the network in charge of coordinating all the transmissions. The CCC scheme contains two phases: the content placement phase, where certain packets of each file are placed into the cache of each user using a predetermined strategy; and the content delivery phase, where the server, upon receiving the specific demands of all users, designs a strategy to broadcast the XOR multiplexing of those requested packets through a shared link. Since then, coded caching has been an active research area and many papers have been written in this aspect, see for example, \cite{ref2}, \cite{ref3}, \cite{ref4}, \cite{ref1}, \cite{ref6}, \cite{ref5}.

The core idea is to design an appropriate content placement strategy such that in the delivery phase various demands of all users can be satisfied with a limited number of multicast transmissions. Each user could recover his requested file with the help of the contents broadcasted in the delivery phase and the contents already stored in his local cache. Assume that we are given $K$ users, $N$ files of unit size (and thus the size of the whole database is $N$) and every user has a local cache of size $M$. In this paper we restrict to the case $N\ge K$. The total transmission amount in the delivery phase is called the {\it rate} of the scheme, denoted as $R$. In order to implement a caching scheme, each file is required to be split into a certain number of packets and we denote this number as $F$. Generally speaking, given $K$, $N$ and $M$, the two parameters $R$ and $F$ are the main evaluating indicators for a caching scheme. The rate $R$ indicates the efficiency of the scheme while $F$ indicates the complexity of the scheme. To evaluate a caching scheme, we assume that the ratio $M/N$ is fixed and analyze the behaviour of $R$ and $F$ as functions of $K$.


For a trivial uncoded caching scheme, each user stores a fraction $M/N$ of each file in his cache. The server broadcasts the left fraction $(1-\frac{M}{N})$ of the requested file to each user according to his request. So the transmission rate for the uncoded caching scheme is just $R_U=K\cdot(1-\frac{M}{N})$, where the first factor $K$ is the rate without caching and the second factor $(1-\frac{M}{N})$ is termed as the local caching gain \cite{AN}. $R_U$ grows linearly with $K$. To implement this scheme it is easy to see that splitting each file into $F_U=N$ packets is enough. $F_U$ is thus a constant independent of $K$.

By simultaneously designing good strategies in the content placement phase and the content delivery phase, the Ali-Niesen scheme could significantly reduce the rate to $R_{AN}=K\cdot(1-\frac{M}{N})\frac{1}{1+KM/N}$. As $K$ grows, the limit of $R_{AN}$ is $\frac{N-M}{M}$, a constant independent of $K$. The rate of the Ali-Niesen scheme has been shown to be optimal. However, to implement the Ali-Niesen scheme, each file has to be split into $F_{AN}$ packets, where $F_{AN}=\binom{K}{KM/N}$ grows exponentially with $K$. This may become infeasible when $K$ is relatively large.

In order to reduce the size of $F_{AN}$, Yan et al. \cite{yan2015placement} reformulate the CCC scheme as a placement delivery array design (or PDA design for simplicity) problem. A PDA shows what is cached for each user in the placement phase and what should be transmitted in the delivery phase.
The problem of constructing a CCC scheme turns into designing a proper PDA for some given parameters. From this perspective Yan et al. propose two kinds of caching schemes. Compared with the Ali-Niesen scheme, $F_{PDA}$ in the scheme of Yan et al. is significantly smaller than $F_{AN}$, while the rate $R_{PDA}$ only suffers from a slight sacrifice compared with $R_{AN}$. However, $F_{PDA}$ still grows exponentially with $K$.

From the schemes above we can see, intuitively, there is a tradeoff between the two parameters $F$ and $R$. The ultimate objective is to design Pareto-optimal CCC schemes with respect to these two related parameters. In particular, if we want to construct a CCC scheme with constant rate $R$, then the Ali-Niesen scheme and the scheme of Yan et al. show that such schemes exist when $F$ grows exponentially with $K$. A natural question is to consider what may be the smallest possible $F$ such that a scheme with constant rate exists. This is the main motivation of the current paper.

In this paper, we follow the steps of \cite{yan2015placement}. We find that the concept of PDA has a natural correspondence with an important problem in extremal combinatorics. We will show that a PDA exists if and only if a corresponding linear and (6, 3)-free 3-uniform 3-partite hypergraph exists. From this point of view, it is very intuitive and easy to understand the essence on how to construct a CCC scheme. By using the well-known (6, 3)-theorem in extremal combinatorics, we first illustrate that constant rate CCC schemes with $F$ growing linearly with $K$ do not exist. Then we present two infinite classes of constructions (or hypergraphs satisfying the corresponding constraints), one from the union of the disjoint subsets including the Ali-Niesen scheme \cite{AN} as a special case, and the other from the extended $q$-ary sequences including the scheme of Yan et al. as a special case (with only a slight negligible difference). By analyzing our schemes, we show that constant rate CCC schemes with $F$ growing sub-exponentially with $K$ do exist.

The rest of this paper is organized as follows. In Section II we briefly review the CCC scheme built in \cite{AN} and the PDA design introduced in \cite{yan2015placement}. In Section III we present our hypergraph model for the CCC scheme, and then its equivalence with the PDA design is established. We will apply the famous (6, 3)-theorem in extremal combinatorics to show that constant rate CCC schemes with $F$ growing linearly with $K$ do not exist. In Section IV we introduce our first general hypergraph construction which includes the Ali-Niesen scheme as a special case. In Section V we introduce our second general hypergraph construction which includes the scheme of Yan et al. as a special case (with only a slight negligible difference). Our two schemes show that constant rate CCC schemes with $F$ growing sub-exponentially with $K$ do exist. We compare the performances of some existing CCC schemes in Section VI. Section VII consists of some concluding remarks and two more possible approaches to study the coded caching problem.

\section{The CCC scheme and the PDA design}

We first recall the CCC scheme introduced in \cite{AN}. Consider a caching system with one server connected to $K$ users, denoted as $\ma{K}=\{1,\dots,K\}$, through an error-free shared link. $N$ files ($N\ge K$) denoted as $\{W_1,W_2,\dots,W_N\}$ are stored in the server and assume that every file is of unit size. Each user has a cache of the same size of $M$ units for $0\le M\le N$.
As in many previous papers, the CCC scheme can be illustrated as the following picture.
\begin{center}
\scalebox{0.7}[0.7]{\includegraphics{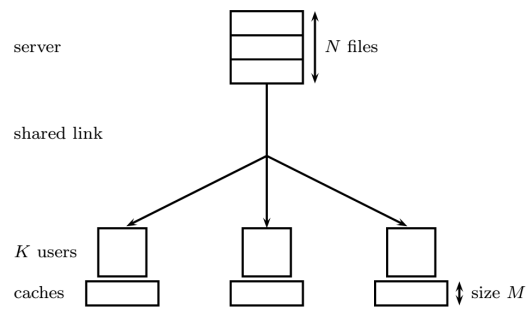}}
\end{center}
\noindent We call this caching scheme a $(K,M,N)$ caching system.

We now briefly review the PDA design problem proposed in \cite{yan2015placement}, which characterizes the CCC scheme in a single array.
A PDA is an array of size $F\times K$, denoted as $\ma{P}=[p_{j,k}]_{F\times K}$, where $F$ is a given integer such that $FM/N$ is an integer. The array consists of a specific symbol $*$ and a set of $S$ integers $\ma{S}=\{1,2,\dots,S\}$. We assume each integer $s\in\ma{S}$ appears at least once in the array. We also denote $\ma{F}=\{1,\ldots,F\}$ and $\ma{N}=\{1,\ldots,N\}$ for simplicity. The following constraints are required:

\begin{enumerate}

\item[C1.] The symbol $*$ appears $Z=FM/N$ times in each column. Therefore each column has $F-Z$ integer entries.

\item[C2.] In each row or each column there do not exist identical integers.

\item[C3.] For any two distinct entries $p_{j_1,k_1}=p_{j_2,k_2}=s\in\ma{S}$, $j_1\neq j_2$ and $k_1\neq k_2$, we have $p_{j_1,k_2}=p_{j_2,k_1}=*$.

\end{enumerate}

We call an array satisfying the constraints above as a $(K,F,Z,S)$-PDA. If each integer in $\ma{S}$ appears $g$ times in $\ma{P}$, we further call this array $g$-regular, which is denoted as a $g$-$(K,F,Z,S)$ PDA. As mentioned earlier, the coded caching scheme has two phases, i.e., the placement phase and the delivery phase. Given a PDA, the corresponding strategy for a caching scheme is as follows.

\begin{enumerate}
  \item Placement phase: Split each file into $F$ packets, that is, $W_i=\{W_{i,j}:j\in\ma{F}\}$. Each user $k\in\ma{K}$ receives the following packets in his cache: $$Z_k=\{W_{i,j}:p_{j,k}=*,i\in\ma{N}\}.$$
        It can be checked that each user has a cache of size $Z\cdot\frac{1}{F}\cdot N=M$.
  \item Delivery phase: Once the server receives the request $d=(d_1,\dots,d_K)$, where $d_k\in\ma{N}$ is the index of the requested file of the user $k$, the server broadcasts the following XOR multiplexing of packets at the time slot $s$ for each $s\in\ma{S}$:
$$ \bigoplus_{p_{j,k}=s, j\in\ma{F}, k\in\ma{K}} W_{d_k,j}. $$
\end{enumerate}

The decoding algorithm for each user is as follows. For a user $k\in\ma{K}$ requesting a certain file $W_{d_k}$ with $d_k\in\ma{N}$, he already has known $\{W_{d_k,j}:p_{j,k}=*\}$ due to the placement phase. To recover $W_{d_k}$, it suffices to decode the unknown packets $\{W_{d_k,j}:p_{j,k}\in\ma{S}\}$. Note that for each $s\in\ma{S}$, in the broadcasted message $\bigoplus_{p_{j,k}=s, j\in\ma{F}, k\in\ma{K}} W_{d_k,j}$, by the constraint C3 the user $k$ knows all the packets $\{W_{d_{k'},j}:p_{j,k'}=s,~k'\neq k\}$ in his cache at the placement phase. Then the unknown value $W_{d_k,j},~p_{j,k}=s$ can be easily computed by substracting $\bigoplus_{p_{j,k'}=s, j\in\ma{F}, k'\in\ma{K},k'\neq k} W_{d_{k'},j}$ from $\bigoplus_{p_{j,k}=s, j\in\ma{F}, k\in\ma{K}} W_{d_k,j}$.
Thus by simple calculation each user will recover his requested file. Therefore the caching scheme works.

\begin{example}\label{example}
   As an example, we present a (2,4,2,2)-PDA for a (2,1,2) CCC scheme. It is not hard to check that the following array is a (2,4,2,2)-PDA.
   \begin{center}
   $\ma{P}_{4,2}=
   \left(
     \begin{array}{cc}
       * & 1 \\
       1 & * \\
       * & 2 \\
       2 & * \\
     \end{array}
   \right)
   $
   \end{center}

   Assume that we are given two files $W_1$ and $W_2$. Divide each file into $F=4$ packets such that $W_1=\{W_{1,1},W_{1,2},W_{1,3},W_{1,4}\}$ and $W_2=\{W_{2,1},W_{2,2},W_{2,3},W_{2,4}\}$. Let $Z_1$ and $Z_2$ be the caches for the two users. In the placement phase, the first user stores $Z_1=\{W_{1,1},W_{1,3},W_{2,1},W_{2,3}\}$ and the second user stores $Z_2=\{W_{1,2},W_{1,4},W_{2,2},W_{2,4}\}$. According to the PDA described above, the contents broadcasted by the server in the delivery phase are illustrated in Table I.

   \begin{table}[h]
    \centering
    \begin{tabular}{|c|c|c|}
      \hline
      Request $d$ & Time slot 1 & Time slot 2 \\\hline
      (1,1) & $W_{1,2}\bigoplus W_{1,1}$ & $W_{1,4}\bigoplus W_{1,3}$ \\\hline
      (1,2) & $W_{1,2}\bigoplus W_{2,1}$ & $W_{1,4}\bigoplus W_{2,3}$ \\\hline
      (2,1) & $W_{2,2}\bigoplus W_{1,1}$ & $W_{2,4}\bigoplus W_{1,3}$ \\\hline
      (2,2) & $W_{2,2}\bigoplus W_{2,1}$ & $W_{2,4}\bigoplus W_{2,3}$ \\
      \hline
    \end{tabular}
        \caption{Delivery phase in Example \ref{example}}
    \end{table}

   To explain the decoding algorithm, suppose that $d=(1,2)$. The first user can recover $W_1$ by decoding $W_{1,2}$ and $W_{1,4}$. Note that $W_{1,2}$ and $W_{1,4}$ can be computed by subtracting $W_{2,1}$ from $W_{1,2}\bigoplus W_{2,1}$
   and $W_{2,3}$ from $W_{1,4}\bigoplus W_{2,3}$, respectively. The second user can recover $W_2$ by decoding $W_{2,1}$ and $W_{2,3}$. Note that $W_{2,1}$ and $W_{2,3}$ can be computed by subtracting $W_{1,2}$ from $W_{1,2}\bigoplus W_{2,1}$
   and $W_{1,4}$ from $W_{1,4}\bigoplus W_{2,3}$, respectively. The decoding algorithms for other requests are similar.
\end{example}

Whatever the request is, the caching scheme represented by the PDA will broadcast $S$ packets, where each packet is of size $1/F$. Thus the rate of this scheme is $R=S/F$. Moreover, if each file is divided into $F$ packets in the placement phase, then such a scheme is termed as an $F$-division scheme. In \cite{yan2015placement}, the authors proved the following fundamental theorem for the PDA design.

\begin{theorem}[\cite{yan2015placement}]\label{PDA}

    An $F$-division caching scheme for a $(K,M,N)$ caching system can be realized by a $(K,F,Z,S)$-PDA $\ma{P}=[p_{j,k}]_{F\times K}$ with $Z/F=M/N$. Each user can decode his requested file correctly for any request $d$ at the rate $R=S/F$.
\end{theorem}

\begin{remark}\label{euqiofAN} The Ali-Niesen scheme in \rm{\cite{AN}} is equivalent to a $(K,\binom{K}{t},\binom{K-1}{t-1},\binom{K}{t+1})$-PDA, where $t=KM/N$ is an integer.
\end{remark}

\begin{remark}\label{equivtang}The first scheme of Yan et al. introduced in \rm{\cite{yan2015placement}} is a $(q(m+1),q^m,q^{m-1},q^{m+1}-q^m)$-PDA.
\end{remark}

Summing up the above, the problem of designing a CCC scheme turns into designing a PDA under given parameters.
Throughout the paper we consider caching schemes where the ratio $M/N$ is fixed, $N\ge K$ and $K$ goes into infinity. We analyze a caching scheme by focusing on the behaviour of $F$ and $R$ with respect to $K$.

\section{The hypergraph model} \label{hypergraph}

Now we turn to a hypergraph perspective towards the CCC scheme or the PDA design problem. We first introduce some necessary definitions. When speaking about a hypergraph we mean a pair $\ma{G}=(V(\ma{G}),E(\ma{G}))$, where the edge set $E(\ma{G})$ is identified as a collection of subsets of the vertex set $V(\ma{G})$. $\ma{G}$ is said to be linear if for all distinct $A,B\in E(\ma{G})$ it holds that $|A \cap B|\le1$. We say $\ma{G}$ is $r$-uniform if $|A|=r$ for all $A\in E(\ma{G})$.

An $r$-uniform hypergraph $\ma{G}$ is $r$-partite if its vertex set $V(\ma{G})$ can be colored in $r$ colors in such a way that no edge of $\ma{G}$ contains two vertices of the same color. In such a coloring, the color classes of $V(\ma{G})$, i.e., the sets of all vertices of the same color, are called parts of $\ma{G}$. In this paper we mainly concern $3$-uniform $3$-partite hypergraphs with three parts $\ma{F}$, $\ma{K}$, $\ma{S}$ such that $|\ma{F}|=F$, $|\ma{K}|=K$ and $|\ma{S}|=S$.

Brown, Erd{\H{o}}s and S{\'o}s \cite{bes1,bes2} introduce the function $f_r(n,v,e)$ to denote the maximum number of edges in an $r$-uniform hypergraph on $n$ vertices which does not contain $e$ edges spanned by $v$ vertices. In other words, in such hypergraphs the size of the union of arbitrary $e$ edges is at least $v+1$. These hypergraphs are called $G(v,e)$-free (or simply $(v,e)$-free). The famous (6, 3)-theorem of Ruzsa and Szemer{\'e}di \cite{RS} points out that

\begin{lemma}\label{6,3}
    \begin{equation*}
            n^{2-o(1)}<f_3(n,6,3)=o(n^2).
    \end{equation*}
\end{lemma}

\noindent This lemma indicates that if a 3-uniform hypergraph on $n$ vertices is (6, 3)-free, then the magnitude of the number of edges can not be linear with $n^2$.

Recall the definition of $\ma{F}$, $\ma{K}$, $\ma{S}$ introduced in Section II. The following observation is the starting point of our approach. A PDA is actually an $F\times K$ array $\ma{P}$ whose entry locates in an alphabet of size $S+1$ (the ``plus one" corresponds to the symbol *). Let us pick a linear 3-uniform 3-partite hypergraph $\ma{H}$ with three parts $\ma{F}$, $\ma{K}$, $\ma{S}$ such that $|\ma{F}|=F$, $|\ma{K}|=K$ and $|\ma{S}|=S$. We connect an edge $\{j,k,s\}$ for $j\in\ma{F},~k\in\ma{K},~s\in\ma{S}$ if and only if the entry in the $j$-th row and the $k$-th column of $\ma{P}$ is exactly $s\in\ma{S}$. Then this hypergraph $\ma{H}$ is uniquely determined by the array $\ma{P}$ and vice versa (in the opposite direction, suppose that we are given a linear 3-uniform 3-partite hypergraph $\ma{H}$ with parts $\ma{F}$, $\ma{K}$, $\ma{S}$, then we can construct a corresponding $F\times K$ array $\ma{P}$ with entries belonging to $\ma{S}\cup\{*\}$, see Theorem \ref{equivalence} below for the details). Then $\ma{H}$ is called the hypergraph defined by the PDA $\ma{P}$. It is easy to check that the number of edges in the hypergraph equals the number of integer entries in the PDA. One can compute that $|E(\ma{H})|=K(F-Z)=KF(1-\frac{M}{N})$.

An important reason why we use the hypergraph perspective is that the three constraints for the PDA design can be easily translated into the corresponding constraints for the hypergraph. The following theorem establishes the equivalence between the PDAs and a class of (6, 3)-free hypergraphs.

\begin{theorem}\label{equivalence}
    A $(K,F,Z,S)$-PDA satisfying constraints C1, C2 and C3 exists if and only if the hypergraph defined by the PDA is a linear and
    (6, 3)-free 3-uniform 3-partite hypergraph with three parts $\ma{F}$, $\ma{K}$, $\ma{S}$ such that $|\ma{F}|=F$, $|\ma{K}|=K$ and $|\ma{S}|=S$. Furthermore, each vertex $k\in\ma{K}$ is incident with exactly $F-Z$ edges.
\end{theorem}

\begin{proof}
    Let $\ma{H}$ denote the hypergraph defined by the $(K,F,Z,S)$-PDA $\ma{P}$. On one hand, to prove the ``only if" part, it suffices to verify that $\ma{H}$ satisfies the constraints proposed in the theorem.
    \begin{enumerate}
      \item It is easy to check that $\ma{H}$ is a 3-uniform 3-partite hypergraph with three parts $\ma{F}$, $\ma{K}$, $\ma{S}$.
      \item The linearity of $\ma{H}$ can be derived from the constraint C2. Let us check it case by case. First, we do not have two edges of the form $\{j,k,s\}$ and $\{j,k,s'\}$ since $p_{j,k}$ is well-defined and has a unique value. Second, we do not have two edges of the form $\{j,k,s\}$ and $\{j,k',s\}$ since otherwise $p_{j,k}=p_{j,k'}=s$, which is forbidden by C2. Finally, we do not have two edges of the form $\{j,k,s\}$ and $\{j',k,s\}$ since otherwise $p_{j,k}=p_{j',k}=s$, which is also forbidden by C2.
      \item Each vertex $k\in\ma{K}$ is incident with exactly $F-Z$ edges, since by C1 each column of $\ma{P}$ contains precisely $F-Z$ integers, which induce $F-Z$ edges in $\ma{H}$.
      \item $\ma{H}$ is (6, 3)-free. In other words, the union of arbitrary three edges of $\ma{H}$ contains at least seven vertices. Assume otherwise, suppose there are three edges spanned by at most six vertices. If the number of vertices is at most five, then one can easily deduce that there must exist two edges having two common vertices, violating the linearity of the hypergraph. It remains to consider the case that three edges are spanned by exactly six vertices.

          Consider how the six vertices are chosen from the three parts. We say that these vertices are divided into the form $a/b/c$ if we choose $a$ vertices from the first part, $b$ from the second and the rest $c$ from the remaining part. If they are divided into $4/1/1$ or $3/2/1$ (permutations included), then the contradiction can also be deduced from the linearity of the hypergraph.
          For the case $2/2/2$, we can always denote these vertices as $j_1,j_2,k_1,k_2,s_1,s_2$. Suppose we have three edges, then without loss of generality assume that $s_1$ appears in two edges, say $\{j_1,k_1,s_1\}$ and $\{j_2,k_2,s_1\}$. Then the possible candidate for the third edge is either $\{j_1,k_2,s_2\}$ or $\{j_2,k_1,s_2\}$. However, from C3, $p_{j_1,k_1}=p_{j_2,k_2}=s_1$ implies that $p_{j_1,k_2}=p_{j_2,k_1}=*$. So there are no edges of the form $\{j_1,k_2,s_2\}$ or $\{j_2,k_1,s_2\}$ for any $s_2\in\ma{S}$. Thus we do not have three edges which are spanned by six vertices.
    \end{enumerate}

    On the other hand, to prove the ``if" part, the following observation is crucial. If we are given a linear and (6, 3)-free 3-uniform 3-partite hypergraph $\ma{H}$ with parts $\ma{F}$, $\ma{K}$, $\ma{S}$, then we can construct a corresponding $F\times K$ array $\ma{P}$ whose entry belongs to $\ma{S}\cup\{*\}$.
    The value in the $j$-th row and the $k$-th column of $\ma{P}$ is $s\in\ma{S}$ if $\{j,k,s\}$ forms an edge of $\ma{H}$ and * otherwise. One can see that we do not have two edges of the form $\{j,k,s\}$ and $\{j,k,s'\}$ by the linearity constraint, thus $p_{j,k}$ is well-defined. It is routine to verify that $\ma{P}$ satisfies C1 and C2. And C3 can be verified by contradiction.
\end{proof}

To see the power and the clarity of the hypergraph perspective towards the PDA design problem, we first prove the following theorem, which is an immediate consequence of Lemma \ref{6,3} and Theorem \ref{equivalence}.

\begin{theorem}\label{linearpda}
    If $R=S/F$ and $M/N$ are both given constants independent of $K$, then for sufficiently large $K$, a $(K,F,Z,S)$-PDA where $F$ grows linearly with $K$ does not exist.
\end{theorem}

\begin{proof}
For sufficiently large $K$, assume that a $(K,F,Z,S)$-PDA with $F=\Theta(K)$ does exist for some $K,~F,~Z,~S$ satisfying the conditions of this theorem. Recall that $Z/F=M/N$. Consider the hypergraph $\ma{H}$ defined by such a PDA, it holds that $|V(\ma{H})|=|\ma{F}|+|\ma{K}|+|\ma{S}|=F+K+S=\Theta(K)+K+RF=\Theta(K)$ and $|E(\ma{H})|=K(F-Z)=KF(1-M/N)=\Theta(K^2)=\Theta(|V(\ma{H})|^2)$.

On the other hand, by Theorem \ref{equivalence} we know $\ma{H}$ is (6, 3)-free and hence by Lemma \ref{6,3} we have $|E(\ma{H})|=o(|V(\ma{H})|^2)$, a contradiction.
\end{proof}

As we have mentioned in the introductory section, the motivation of this paper is to consider what may be the smallest possible $F$ such that a scheme with constant rate $R$ exists. Theorem \ref{linearpda} actually offers a lower threshold, that is, constant rate schemes with $F$ growing linearly with $K$ do not exist. The schemes of Ali-Niesen and Yan et al. indicate that $F$ growing exponentially with $K$ suffices. Could we lower the magnitude of $F$?  As a first step, we will show that constant rate schemes with $F$ growing sub-exponentially with $K$ do exist, by offering two such schemes in the next two sections. The schemes are introduced directly in the hypergraph perspective, that is, we are actually constructing linear and (6, 3)-free 3-uniform 3-partite hypergraphs.

\section{Constructions from the union of disjoint subsets} \label{scheme1}

In this section we present our first scheme, including the Ali-Niesen scheme as a special case.

\fbox{%
  \parbox{\textwidth}{%
Scheme 1:

Let $n,a,b$ be positive integers satisfying $n\ge a+b$. Let $[n]=\{1,2,\dots,n\}$ and let $\binom{[n]}{a}=\{A\s[n]:|A|=a\}$ denote the collection of subsets of $[n]$ of size $a$. Then we construct a 3-uniform 3-partite hypergraph $\ma{H}_1$ as follows. Let $V_1,~V_2,~V_3$ be three parts of $V(\ma{H}_1)$ such that $V_1=\binom{[n]}{a}$, $V_2=\binom{[n]}{b}$ and $V_3=\binom{[n]}{a+b}$.
Three vertices $A\in V_1,~B\in V_2,~C\in V_3$ form an edge $\{A,B,C\}$ if and only if $|A|=a$, $|B|=b$, $|C|=c$ and $A\cup B=C$.
  }%
}

\begin{theorem}\label{scheme1the}
$\ma{H}_1$ is a linear and (6, 3)-free 3-uniform 3-partite hypergraph.
\end{theorem}

\begin{proof}
It is routine to check that this hypergraph is 3-uniform and 3-partite. For an edge $\{A,B,C\}$, every two vertices uniquely determine the third one, so the linearity is straightforward. Assume that there are three edges spanned by six vertices. By the proof of Theorem \ref{equivalence}, we only need to consider the case where the six vertices are chosen averagely from the three vertex parts. For every six vertices $A,A',B,B',C,C'$, if they induce three edges, then without loss of generality we can assume that there exist two edges $\{A,B,C\}$ and $\{A',B',C\}$. However, $A\cup B = A'\cup B' =C$ and $A\neq A'$ indicates that $A\cap B'\neq\emptyset$ and $A'\cap B\neq\emptyset$, so $|A\cup B'|<a+b$ and $|A'\cup B|<a+b$. Thus we do not have edges of the form $\{A,B',C'\}$ or $\{A',B,C'\}$ for any $C'\in V_3$. Therefore we do not have three edges which are spanned by six vertices.\end{proof}

\begin{theorem}
    For every three positive integers $a,~b,~n$ such that $a+b\le n$, there exists an $\binom{a+b}{a}$-regular $(\binom{n}{b},\binom{n}{a},\binom{n}{a}-\binom{n-b}{a},\binom{n}{a+b})$-PDA.
\end{theorem}

\begin{proof}
    Take $\ma{F}=V_1,~\ma{K}=V_2$ and $\ma{S}=V_3$. Then by the construction of Scheme 1 it is easy to see that every vertex in $\ma{K}$ is incident with exactly $\binom{n-b}{a}$ edges. Therefore, by Theorems \ref{equivalence} and \ref{scheme1the} we can conclude that there exists a $(K,F,Z,S)$-PDA with $K=\binom{n}{b},~F=\binom{n}{a},~Z=\binom{n}{a}-\binom{n-b}{a}$ and $S=\binom{n}{a+b}$. Furthermore, it is $\binom{a+b}{a}$-regular since for any $C\in\binom{[n]}{a+b}$, it holds that $|\{(A,B):A\in\binom{[n]}{a},~B\in\binom{[n]}{b},~A\cup B=C\}|=\binom{a+b}{a}$.
\end{proof}

One can see that if we choose $b=1$, $n=K$ and $a=KM/N$, then Scheme 1 obviously includes the
Ali-Niesen scheme as a special case. Actually, our hypergraph perspective reveals the essential structure of the Ali-Niesen scheme.

For the general case, we have $R=S/F=\binom{n}{a+b}/\binom{n}{a}$, $F=\binom{n}{a}$, $M/N=Z/F=1-\binom{n-b}{a}/\binom{n}{a}$ and $K=\binom{n}{b}$. In general, it is not easy to measure the performance of this scheme since we can hardly express $R$ or $F$ as functions of $K$. However, $R$ is far better than the uncoded scheme with $R_{U}=K(1-\frac{M}{N})={n\choose b}{{n-b} \choose a}/{n\choose a}$, since $R_{U}/R=\binom{a+b}{a}\gg1$. So the new scheme does make sense. Actually, by choosing proper parameters, our Scheme 1 may contribute a lot of constant rate CCC schemes with $F$ growing sub-exponentially with $K$. An example is as follows.

\begin{remark}\label{scheme1b=2}If we take $b=2$ and then we obtain an $(\binom{n}{2},\binom{n}{a},\binom{n}{a}-\binom{n-2}{a},\binom{n}{a+2})$-PDA. If we set $n=\lambda a$ for some constant $\lambda>1$, then it holds that $R=S/F=\binom{n}{a+2}/\binom{n}{a}\thickapprox(\lambda-1)^2$ and $M/N=Z/F=(\binom{n}{a}-\binom{n-2}{a})/\binom{n}{a}\thickapprox\fr{2\lambda-1}{\lambda^2}$ and by Stirling's formula we have
$$F=\binom{n}{\lambda^{-1} n}=\fr{1+o(1)}{\sqrt{2\pi\lambda^{-1}(1-\lambda^{-1})n}}\cdot2^{n H(\lambda^{-1})}=\ma{O}(K^{-1/4}\cdot2^{\sqrt{2K}H(\lambda^{-1})}),$$
where $H(x)=-x\log_2x-(1-x)\log_2(1-x)$ for $0<x<1$ is the binary entropy function. It is easy to see that under such choice of parameters, $R$ and $M/N$ are both constants independent of $K$ and $F$ grows sub-exponentially with $K$.
\end{remark}

\section{Constructions from the extended $q$-ary sequences} \label{scheme2}

In this section we present our second scheme, including the scheme of Yan et al. as a special case (with only a slight negligible difference).

\fbox{%
  \parbox{\textwidth}{%
  Scheme 2:

Let $q,m,t$ be positive integers with $t\le m$ and let $\mathbb{Z}_q=\{0,1,\ldots,q-1\}$. Then we construct a 3-uniform 3-partite hypergraph $\ma{H}_2$ as follows. The first part of the vertices is $W_1=\{A=(a_1,\dots,a_m):a_i\in\mathbb{Z}_q\}$. That is, $W_1$ consists of the $q$-ary vectors of length $m$ and $|W_1|=q^m$. The second part is $W_2=\{B=(\delta_1,\dots,\delta_t,b_{\delta_1},\dots,b_{\delta_t}):1\le \delta_1 < \dots < \delta_t\le m, b_{\delta_i}\in \mathbb{Z}_q\}$. That is, $W_2$ consists of vectors of length $2t$ with the former $t$ coordinates being distinct integers ranging from $1$ to $m$ listed in a strictly increasing order and the latter $t$ coordinates being in $\mathbb{Z}_q$. So $|W_2|={m \choose t}q^t$. Finally $W_3=\{C=(c_1,\dots,c_m,c_{m+1},\dots,c_{m+t}): c_i \in \mathbb{Z}_q\text{ for }1\le i \le m\text{ and }c_{m+j}\in\mathbb{Z}_{q}\backslash\{q-1\}\text{ for }1\le j \le t\}$. Clearly $|W_3|=q^m(q-1)^t$ and for each $1\le j\le t$, it holds that $c_{m+j}+1\not\equiv 0 \pmod q$.

Three vertices $A\in W_1,~B\in W_2,~C\in W_3$ form an edge $\{A,B,C\}$ if and only if the following conditions hold simultaneously. Note that the computations are made in $\mathbb{Z}_q$.

\begin{enumerate}

\item $a_i=c_i$ for $i\notin\{\delta_1,\dots,\delta_t\}$, $1\le i \le m$;

\item $a_{\delta_j}=c_{\delta_j}+c_{m+j}+1$ for $j=1,2,\dots,t$;

\item $b_{\delta_j}=c_{\delta_j}$ for $j=1,2,\dots,t$.

\end{enumerate}
  }%
}

We have the following observation. Compare the first $m$ coordinates of $C$ with $A$. For $i\notin\{\delta_1,\dots,\delta_t\}$, the corresponding entries are identical. As for the other coordinates, since $c_{m+j}\in\mathbb{Z}_{q}\backslash\{q-1\}$ for $1\le j \le t$, then $c_{m+j}+1\neq0$, so the corresponding entries are distinct due to the second constraint. Thus a necessary condition for $A$ and $C$ lying in an edge is that $(a_1,\dots,a_m)$ and $(c_1,\dots,c_m)$ have exactly $t$ distinct entries. Moreover, we have $a_{\delta_j}\neq b_{\delta_j}$ for $j=1,2,\dots,t$.

\begin{theorem}\label{scheme2the}
$\ma{H}_2$ is a linear and (6, 3)-free 3-uniform 3-partite hypergraph.
\end{theorem}

\begin{proof}
It is routine to check that this hypergraph is 3-uniform and 3-partite. It remains to show the linearity and the (6, 3)-free property of $\ma{H}_2$.
\begin{enumerate}
  \item First we show its linearity. We only need to show that for an edge of the form $\{A,B,C\}$, every two vertices (if they do lie in an edge) uniquely determine the last one. For given $A$ and $B$, from the values of $\delta_1,\dots,\delta_t$, the first $m$ coordinates of $C$ can be determined via the first and third constraints. The last $t$ coordinates of $C$ then could be calculated from the second constraint. Similarly, for given $B$ and $C$, solving $A$ is also straightforward. The case of determining $B$ when $A$ and $C$ are given is a little different. From the observation above, $\{\delta_1,\dots,\delta_t\}$ could be determined by comparing $(a_1,\dots,a_m)$ and $(c_1,\dots,c_m)$ and finding out the coordinates where the corresponding entries differ. Then $\{b_{\delta_1},\dots,b_{\delta_t}\}$ could be determined by the third constraint. Therefore, the linearity of $\ma{H}_2$ follows.

  \item For the (6, 3)-property, by the proof of Theorem \ref{equivalence} we only need to consider the case where six vertices are chosen averagely from three vertex parts. For any six vertices $A,A',B,B',C,C'$, suppose they induce three edges, then without loss of generality we assume that we have $\{A,B,C\}$, $\{A',B,C'\}$ and $\{A,B',C'\}$. By using permutations on indices we also can assume that $B=(1,2,\dots,t,b_1,\dots,b_t)$. Then it holds that $C=(b_1,\dots,b_t,c_{t+1},\dots,c_{m+t})$, $A=(b_1+c_{m+1}+1,\dots,b_t+c_{m+t}+1,c_{t+1},\dots,c_m)$ and $C'=(b_1,\dots,b_t,c'_{t+1},\dots,c'_{m+t})$. Comparing the first $t$ entries of $A$ and $C'$, they are all different since $c_{j}+1\neq0$ for $m+1\le j \le m+t$. So to guarantee that $\{A,B',C'\}$ does form an edge we must have $B'=(1,2,\dots,t,b'_1,\dots,b'_t)$ and then $C'=(b'_1,\dots,b'_t,c'_{t+1},\dots,c'_{m+t})$. This infers that $b_i=b'_i$ for $1\le i \le t$, i.e., $B$ and $B'$ are identical, a contradiction.
\end{enumerate}

\end{proof}

\begin{theorem}
    For every three positive integers $q,~t,~m$ with $t\le m$, there exists an $\binom{m}{t}$-regular $(\binom{m}{t}q^t,q^m,q^m-q^{m-t}(q-1)^t,q^m(q-1)^t)$-PDA.
\end{theorem}

\begin{proof}
    Take $\ma{F}=W_1,~\ma{K}=W_2$ and $\ma{S}=W_3$. Consider the number of vertices in $\ma{F}$ adjacent with some given vertex $B\in\ma{K}$. By the construction of Scheme 2 it is easy to see that every vertex in $\ma{K}$ is incident with exactly $q^{m-t}(q-1)^t$ edges, since there are $q$ choices for the coordinates $a_i,~i\not\in\{\delta_1,\ldots,\delta_t\}$ and $q-1$ choices for the coordinates $a_{\delta_j},~1\le j\le t$ (we have $a_{\delta_j}\neq b_{\delta_j}$ for each $1\le j\le t$).
    Therefore, by Theorems \ref{equivalence} and \ref{scheme2the} we can conclude that there exists a $(K,F,Z,S)$-PDA with $K=\binom{m}{t}q^t,~F=q^m,~Z=q^m-q^{m-t}(q-1)^t$ and $S=q^m(q-1)^t$. Furthermore, it is $\binom{m}{t}$-regular since for any given $C\in\ma{S}$ we have $\binom{m}{t}$ choices for $\{\delta_1,...,\delta_t\}$ and once $C$ and $\{\delta_1,...,\delta_t\}$ are fixed, we can determine $A$ and $B$ using the three constraints.
\end{proof}

One can see that if we choose $t=1$, then Scheme 2 induces an $(mq,q^m,q^{m-1},q^m(q-1))$-PDA, which is very close to the first scheme of Yan et al. in Remark \ref{equivtang} (only missing the last $q$ users in their scheme). Actually, we can possibly remove this difference by adding more vertices to $V_2$. We do not attempt to do so since if we fix $q$ and let $m$ approximate infinity, then the difference is negligible.

For the general case, we have $R=S/F=(q-1)^t$, $F=q^m$ and $K=\binom{m}{t}q^t$. We also have $M/N=Z/F=1-(1-1/q)^t$. Let $q$ and $t$ be fixed and let $m$ approximate infinity, it holds that $R$ and $M/N$ are constants independent of $K$. Moreover, if we solve $m$ from $K$ as $m=\Theta((K/q^t)^{1/t})$, then we can write $F$ as $F=\Theta(q^{K^{1/t}/q})$. Obviously, $F$ increases sub-exponentially with $K$ if we set $t\ge 2$.

We further mention another advantage of our hypergraph perspective. In \cite{yan2015placement} the authors actually present two symmetric constructions, i.e., a $(q(m+1),q^m,q^{m-1},q^{m+1}-q^m)$-PDA for $M/N=1/q$ and a $(q(m+1),q^{m+1}-q^m,(q-1)^2q^{m-1},q^m)$-PDA for $M/N=(q-1)/q$. In \cite{yan2015placement} it takes quite a while to state these two constructions separately. However, from our hypergraph perspective, these two constructions are essentially the same: if we are aware of any one of the two, we know the other one immediately. Assume that the first construction is represented by a hypergraph $\ma{G}_1$ with vertex parts $V_1=\ma{F},~V_2=\ma{K}$ and $V_3=\ma{S}$, then the second one actually is represented by a symmetric hypergraph $\ma{G}_2$ with vertex parts $V_1=\ma{S},~V_2=\ma{K}$ and $V_3=\ma{F}$. Therefore, if $\ma{G}_1$ is linear and (6, 3)-free, then so is $\ma{G}_2$. That is, from a given CCC scheme represented by hypergraphs, we may directly obtain a symmetric one by switching the roles of the vertex parts $\ma{F}$ and $\ma{S}$. Thus we have the following corollary.

\begin{corollary}\label{symmetric}
        For every three positive integers $q,~t,~m$, there exists an $(\binom{m}{t}q^t,q^m(q-1)^t,q^m(q-1)^t-q^{m-t}(q-1)^t,q^m)$-PDA with $R=1/(q-1)^t$ and $M/N=Z/F=1-1/q^t$.
\end{corollary}

\begin{remark}\label{scheme2t=2}If we take $t=2$ and then we obtain an
$(\binom{m}{2}q^2,q^m(q-1)^2,q^m(q-1)^2-q^{m-2}(q-1)^2,q^m)$-PDA with $M/N=Z/F=1-1/q^2$.
\end{remark}

\section{Comparison with previous constructions}

In this section, we compare our new constructions with the existing ones. First of all, we summarize all the constructions in Table \ref{comparison}.

\begin{table}[h]
\caption{Comparison of some CCC schemes}\label{comparison}
\centering
\begin{tabular}{|c|c|c|c|c|c|}
  \hline
  & & $K$ & $M/N$ & $F$ & $R$   \\\hline
  Construction 1&A-N \cite{AN} & $K$ & $\fr{1}{q}$ & $\binom{K}{\fr{K}{q}}$ & $\fr{K}{K+q}(q-1)$ \\\hline
  Construction 2&A-N \cite{AN} & $K$ & $\fr{q-1}{q}$ & $\binom{K}{\fr{K}{q}}$ & $\fr{K}{q+K(q-1)}$ \\\hline
  Construction 3&Yan et al. \cite{yan2015placement} & $K$ & $\fr{1}{q}$ & $q^{\fr{K}{q}-1}$ & $q-1$ \\\hline
  Construction 4&Yan et al. \cite{yan2015placement} & $K$ & $\fr{q-1}{q}$ & $(q-1)q^{\fr{K}{q}-1}$ & $\fr{1}{q-1}$ \\\hline
  Construction 5&Scheme 1 & $\binom{n}{b}$ & $1-\fr{\binom{n-b}{a}}{\binom{n}{a}}$ & $\binom{n}{a}$ & $\fr{\binom{n}{a+b}}{\binom{n}{a}}$ \\\hline
  Construction 6&Scheme 1: $b=2,~n=\lambda a$ & $\binom{n}{2}$ & $\thickapprox\fr{2\lambda-1}{\lambda^2}$ & $\binom{n}{\fr{n}{\lambda}}$ & $\thickapprox(\lambda-1)^2$ \\\hline
  Construction 7&Scheme 2 & $\binom{m}{t}q^t$ & $1-(\fr{q-1}{q})^t$ & $q^m$ & $(q-1)^t$ \\\hline
  Construction 8&Scheme 2: symmetric form & $\binom{m}{t}q^t$ & $1-\fr{1}{q^t}$ & $q^m(q-1)^t$ & $\fr{1}{(q-1)^t}$ \\\hline
  Construction 9&Scheme 2: symmetric form, $t=2$ & $\binom{m}{2}q^2$ & $1-\fr{1}{q^2}$ & $q^m(q-1)^2$ & $\fr{1}{(q-1)^2}$ \\\hline
\end{tabular}
\end{table}

Table \ref{comparison} contains a variety of different PDAs. It may be hard to find out the advantages or the disadvantages of them since the parameters are confusable. Therefore, we also present comparisons of some of the schemes under unified parameters. We will take the constructions in Remarks \ref{scheme1b=2} (Construction 6) and \ref{scheme2t=2} (Construction 9) for comparison, since Remark \ref{scheme1b=2} presents a PDA with small $M/N$ and Remark \ref{scheme2t=2} presents a PDA with large $M/N$. For $M/N=1/q,$ we compare Constructions 1, 3 and 6. For Construction 6, we set $K:=\binom{n}{2}$, $\lambda:=2q$, then $n\thickapprox\sqrt{2K}$, $M/N\thickapprox\fr{2\lambda-1}{\lambda^2}\thickapprox1/q$ and $R\thickapprox(2q-1)^2$ and by Remark \ref{scheme1b=2} we have $F\thickapprox\sqrt{\fr{2^{1/2}q^2}{\pi(2q-1)K^{1/2}}}\cdot2^{\sqrt{2K}H(\fr{1}{2q})}$. For $M/N=(q-1)/q$, we compare Constructions 2, 4 and 9. For Construction 9, we set $K:=\binom{m}{2}q^2$, $q^2:=q$, then $m\thickapprox\sqrt{2K/q}$, $M/N=(q-1)/q$, $F\thickapprox\sqrt{q}^{\sqrt{2K/q}}(\sqrt{q}-1)^2= q^{\sqrt{K/2q}}(\sqrt{q}-1)^2$ and $R=1/(\sqrt{q}-1)^2$. The comparisons are listed in Tables \ref{1/q} and \ref{(q-1)/q}.

For $M/N=1/q$, from Table \ref{1/q} one can see that the rate of Construction 6 is almost as large as four times of the square of that of Constructions 1 and 3. But the magnitude of $F$ reduces significantly. Let $q$ be fixed, then $F$ is reduced from $\Omega(q^{K/q})$ to $\ma{O}(q^{\sqrt{8K}/q})$. For $M/N=(q-1)/q$, the advantage of our construction is more remarkable. In Table \ref{(q-1)/q}, the transmission rates of Constructions 2, 4 and 9 are almost the same. However, $F$ is reduced from $\Omega(q^{K/q})$ to $\ma{O}(q^{\sqrt{K/2q}})$.

\begin{table}[h]
\caption{Comparison under $M/N=1/q$}\label{1/q}
\centering
\begin{tabular}{|c|c|c|c|c|}
  \hline
   & $K$ & $M/N$ & $F$ & $R$   \\\hline
  Construction 1 & $K$ &             $\fr{1}{q}$ & $\approx\fr{q}{\sqrt{2\pi K(q-1)}}\cdot q^{\fr{K}{q}}\cdot(\fr{q}{q-1})^{K(1-\fr{1}{q})}$ & $\fr{K}{K+q}(q-1)$ \\\hline
  Construction 3 & $K$ &             $\fr{1}{q}$ & $q^{\fr{K}{q}-1}$ & $q-1$ \\\hline
  Construction 6 & $K$ & $\thickapprox\fr{1}{q}$ & $\thickapprox\sqrt{\fr{2^{1/2}q^2}{\pi(2q-1)K^{1/2}}}\cdot2^{\sqrt{2K}H(\fr{1}{2q})}$ & $\thickapprox(2q-1)^2$ \\\hline
\end{tabular}
\end{table}

\begin{table}[!h]
\caption{Comparison under $M/N=(q-1)/q$}\label{(q-1)/q}
\centering
\begin{tabular}{|c|c|c|c|c|}
  \hline
   & $K$ & $M/N$ & $F$ & $R$   \\\hline
  Construction 2 & $K$ &             $\fr{q-1}{q}$ & $\approx\fr{q}{\sqrt{2\pi K(q-1)}}\cdot q^{\fr{K}{q}}\cdot(\fr{q}{q-1})^{K(1-\fr{1}{q})}$ & $\fr{K}{q+K(q-1)}$ \\\hline
  Construction 4 & $K$ &             $\fr{q-1}{q}$ & $(q-1)q^{\fr{K}{q}-1}$ & $\fr{1}{q-1}$ \\\hline
  Construction 9 & $K$ & $\fr{q-1}{q}$ & $\thickapprox q^{\sqrt{\fr{K}{2q}}}(\sqrt{q}-1)^2$ & $\fr{1}{(\sqrt{q}-1)^2}$ \\\hline
\end{tabular}
\end{table}

\section{Conclusion and related problems} \label{conclusion}

In this paper, we view the problem of constructing a CCC scheme or a PDA design in a hypergraph perspective. The problem gets related to the famous $(6, 3)$-problem in extremal combinatorics. From this point of view, constructing caching schemes turns into constructing linear and (6, 3)-free 3-partite 3-uniform hypergraphs. We offer two schemes, generalizing the Ali-Niesen scheme and the scheme of Yan et al. respectively. The parameters in our schemes are flexible so that they actually contribute two large classes of caching schemes.

What is the smallest possible $F$ such that a constant rate CCC scheme exists? Our constructions indicate that $F$ increasing sub-exponentially with $K$ suffices. As suggested by Theorem \ref{linearpda}, constant rate CCC schemes with $F$ growing linearly with $K$ do not exist. The problem is still not fully understood, for example, we do not know whether constant rate CCC schemes with $F$ growing polynomially with $K$ exist or not. We leave it as an open problem.

{\bf Open Problem:} Let $M/N$ and $R$ be both constants independent of $K$. Find out the minimal $F=f(K)$ such that a $(K,F,Z,S)$-PDA with $S=RF$, $Z=FM/N$ does exist. Especially, prove or disprove that $F$ growing polynomially with $K$ suffices.

Except for the hypergraph perspective, we would like to provide two more interesting approaches to study the PDA design problem.

{\bf Partial Latin square with the Blackburn property:} A Latin square is an $n\times n$ matrix $\ma{L}$ filled with $n$ different symbols $\{1,\ldots,n\}$, each occurring exactly once in each row and exactly once in each column. A partial Latin square is a submatrix formed by several rows and columns of $\ma{L}$. We say that a partial Latin square $\ma{P}$ has the Blackburn property if whenever two distinct cells $\ma{P}_{a,b}$ and $\ma{P}_{c,d}$ are occupied by the same symbol, the opposite corners $\ma{P}_{a,d}$ and $\ma{P}_{b,c}$ are blank. We further call this partial Latin square regular if each column has the same number of symbols. The problem of filling as many cells without violating this property as possible is posed by Blackburn \cite{Blackburn} and studied by Wanless \cite{latin}. One can verify that the definition of a PDA array is indeed equivalent to that of a regular partial Latin square with the Blackburn property. For example, one can argue that the constraints C1, C2 and C3 are equivalent to the regular property, the Latin property and the Blackburn property, respectively.

{\bf Strong edge coloring for bipartite graphs:} A strong edge-coloring of a graph $\ma{G}$ is an edge-coloring in which every color class is an induced matching; that is, any two vertices belonging to distinct edges with the same color are not adjacent. The strong chromatic index $S(\ma{G})$ is the minimum number of colors in a strong edge-coloring of $\ma{G}$. If we are given a PDA $\ma{P}$, then we can construct a bipartite graph with vertex sets $\ma{F}$ and $\ma{K}$ such that $j\in\ma{F}$ and $k\in\ma{K}$ are connected by an edge if and only if $p_{j,k}\in\ma{S}$. We color this edge by a color $s\in\ma{S}$ if $p_{j,k}=s$. By the constraints C2 and C3 one can verify that such a coloring is a strong edge coloring. Then $S(\ma{G})$ is the minimal $S$ such that a PDA does exist. See \cite{coloring} for an introduction of the strong edge coloring problem.
\bibliographystyle{IEEEtrans}
\bibliography{caching}

\end{document}